\documentclass[11pt,a4paper,twoside]{article}

\usepackage{amsmath,amsfonts,amssymb,amsthm}
\usepackage{algorithm,algorithmic}
\usepackage{graphicx}
\usepackage{prelim2e}
\usepackage{subfig}

\usepackage{color}
\definecolor{DarkBlue}{rgb}{0,0,0.75}
\definecolor{DarkRed}{rgb}{0.75,0,0}
\newcommand{\new}{}

\theoremstyle{plain}
  \newtheorem{theorem}{Theorem}
  \newtheorem*{theo}{Theorem}
  \newtheorem{corollary}[theorem]{Corollary}
  \newtheorem{lemma}[theorem]{Lemma}
  \newtheorem{proposition}[theorem]{Proposition}
\theoremstyle{definition}
  \newtheorem*{conjecture}{Conjecture}

  \newtheorem*{assumption}{Assumption A}
  \newtheorem*{convention}{Convention}
\theoremstyle{remark}
  \newtheorem*{remark}{Remark}
  \newtheorem{example}{Example}
\newcommand{\mdef}[1]{\textsl{#1}}

\newcommand{\abs}[1]{|#1|}
\newcommand{\set}[1]{\{#1\}}
\newcommand{\R}{\mathbb{R}}

\renewcommand{\vec}[1]{\boldsymbol{#1}}
\newcommand{\prv}{\vec{\pi}}
\newcommand{\prvt}{\widetilde{\prv}}
\newcommand{\vun}{\vec{1}}
\newcommand{\ens}[1]{\mathcal{#1}}
\newcommand{\intr}{{\ens I}}
\newcommand{\ext}{{\overline{\intr}}}
\newcommand{\out}{\mathrm{out(\intr)}}
\newcommand{\inn}{\mathrm{in(\intr)}}

\newcommand{\ch}[1]{\gets{#1}}
\newcommand{\cht}[1]{\widetilde\gets{#1}}
\newcommand{\dg}[1]{d_{#1}}
\newcommand{\dgt}[1]{\widetilde{d}_{#1}}
\newcommand{\maxV}{\ens V}
\newcommand{\maxVt}{\widetilde{\ens V}}

\DeclareMathOperator*{\argmin}{argmin}
\DeclareMathOperator*{\argmax}{argmax}

\title{Maximizing PageRank via outlinks}

\author{Cristobald de~Kerchove \and Laure~Ninove \and Paul Van~Dooren}

\date{\small{\emph{CESAME,
Universit{\'e} catholique de Louvain,\\
Avenue~Georges~Lema\^itre~4--6,
B-1348~Louvain-la-Neuve, Belgium\\
\{c.dekerchove, laure.ninove, paul.vandooren\}@uclouvain.be}}}

\begin{document}
\maketitle

\begin{abstract}
We analyze linkage strategies for a set $\intr$ of webpages for which the webmaster wants to maximize the sum of Google's PageRank scores. The webmaster can only choose the hyperlinks \emph{starting} from the webpages of $\intr$ and has no control on the hyperlinks from other webpages. We provide an optimal linkage strategy under some reasonable assumptions.

\vspace{1ex}
\noindent\emph{Keywords:} PageRank, Google matrix, Markov chain, Perron vector, Optimal linkage strategy

\vspace{1ex}
\noindent\emph{AMS classification:} 15A18, 15A48, 15A51, 60J15, 68U35
\end{abstract}

\section{Introduction}

PageRank, a measure  of webpages' relevance introduced by Brin and Page, is at the heart of the well known search engine Google~\cite{BP98,PBMW98}. Google classifies the webpages according to the pertinence scores given by PageRank, which are computed from the graph structure of the Web. A page with a high PageRank will appear among the first items in the list of pages corresponding to a particular query.

If we look at the popularity of Google, it is not surprising that some webmasters want to increase the PageRank of their webpages in order to get more visits from websurfers to their website. Since PageRank is based on the link structure of the Web, it is therefore useful to understand how addition or deletion of hyperlinks influence it.

Mathematical analysis of PageRank's sensitivity with respect to perturbations of the matrix describing the webgraph is a topical subject of interest (see for instance~\cite{AL06,BGS05,Kir06,LM04,LM06book,LM05} and the references therein). Normwise and componentwise conditioning bounds~\cite{Kir06} as well as the derivative~\cite{LM04,LM06book} are used to understand the sensitivity of the PageRank vector. It appears that the PageRank vector is relatively insensitive to small changes in the graph structure, at least when these changes concern webpages with a low PageRank score~\cite{BGS05,LM04}. One could think therefore that trying to modify its PageRank via changes in the link structure of the Web is a waste of time. However, what is important for webmasters is not the values of the PageRank vector but the \emph{ranking} that ensues from it. Lempel and Morel~\cite{LM05} showed that PageRank is not rank-stable, i.e.\ small modifications in the link structure of the webgraph may cause dramatic changes in the ranking of the webpages. Therefore, the question of how the PageRank of a particular page or set of pages could be increased--even slightly--by adding or removing links to the webgraph remains of interest.

As it is well known~\cite{AL04,IW06}, if a hyperlink from a page $i$ to a page $j$ is added, without no other modification in the Web, then the PageRank of $j$ will increase. But in general, you do not have control on the \emph{inlinks} of your webpage unless you pay another webmaster to add a hyperlink from his/her page to your or you make an \emph{alliance} with him/her by trading a link for a link~\cite{BYCL05,GGM05}. But it is natural to ask how you could modify your PageRank by yourself. This leads to analyze how the choice of the \emph{outlinks} of a page can influence its own PageRank. Sydow~\cite{Syd05} showed via numerical simulations that adding well chosen outlinks to a webpage may increase significantly its PageRank ranking. Avrachenkov and Litvak~\cite{AL06} analyzed theoretically the possible effect of new outlinks on the PageRank of a page and its neighbors. Supposing that a webpage has control only on its outlinks, they gave the optimal linkage strategy for this single page. Bianchini et al.~\cite{BGS05} as well as Avrachenkov and Litvak in~\cite{AL04} consider the impact of links between web communities (websites or sets of related webpages), respectively on the sum of the PageRanks and on the individual PageRank scores of the pages of some community. They give general rules in order to have a PageRank as high as possible but they do not provide an optimal link structure for a website.

\medskip
Our aim in this paper is to find a generalization of Avrachenkov--Litvak's optimal linkage strategy~\cite{AL06} to the case of \emph{a website with several pages}. \new{We consider a given set of pages and suppose we have only control on the \emph{outlinks} of these pages. We are interested in the problem of \emph{maximizing the sum of the PageRanks} of these pages.}

\medskip
Suppose $\mathcal{G}=(\ens N,\ens E)$ be the webgraph, with a set of nodes $\ens N=\set{1,\dots,n}$ and a set of links $\ens E\subseteq\ens N\times \ens N$. For a subset of nodes $\intr\subseteq\ens N$, we define
\begin{align*}
  \ens E_\intr&=\set{(i,j)\in\ens E\colon i,j\in\intr} \text{ the set of internal links},\\
  \ens E_\out&=\set{(i,j)\in\ens E\colon i\in\intr,j\notin\intr}
    \text{ the set of external outlinks},\\
  \ens E_\inn&=\set{(i,j)\in\ens E\colon i\notin\intr,j\in\intr}
    \text{ the set of external inlinks},\\
  \ens E_\ext&=\set{(i,j)\in\ens E\colon i,j\notin\intr} \text{ the set of external links}.
\end{align*}

\medskip
\new{If we do not impose any condition on $\ens E_\intr$ and $\ens E_\out$, the problem of maximizing the sum of the PageRanks of pages of $\intr$ is quite trivial and does not have much interest (see the discussion in Section~\ref{sec:optimal-linkage-strategy}). Therefore, when characterizing optimal link structures, we will make the following \emph{accessibility assumption}: every page of the website must have an access to the rest of the Web.}

\medskip
Our first main result concerns the \emph{optimal outlink structure} for a given website. In the case where the subgraph corresponding to the website is strongly connected, Theorem~\ref{thm:out-opt} can be particularized as follows.
\begin{theo}
  Let $\ens E_\intr$, $\ens E_\inn$ and $\ens E_\ext$ be given.
  Suppose that the subgraph $(\intr,\ens E_\intr)$ is strongly connected and
  $\ens E_\intr\neq\emptyset$.
  Then every optimal outlink structure $\ens E_\out$
  is to have only one outlink to a particular page outside of $\intr$.
\end{theo}

\medskip
We \new{are} also \new{interested in} the optimal \emph{internal} link structure for a website. In the case where there is a unique leaking node in the website, that is only one node linking to the rest of the web, Theorem~\ref{thm:intr-opt} can be particularized as follows.
\begin{theo}
  Let $\ens E_\out$, $\ens E_\inn$ and $\ens E_\ext$ be given.
  Suppose that there is only one leaking node in $\intr$.
  Then every optimal internal link structure $\ens E_\intr$ is composed of
  together with every possible backward link.
\end{theo}

\medskip
Putting together Theorems~\ref{thm:out-opt} and~\ref{thm:intr-opt}, we get in Theorem~\ref{thm:website-opt} the \emph{optimal link structure} for a website. This optimal structure is illustrated in Figure~\ref{fig:optimal-link-structure}.
\begin{theo}
  Let $\ens E_\inn$ and $\ens E_\ext$ be given.
  Then, for every optimal link structure,
  $\ens E_\intr$ is composed of a forward chain of links
  together with every possible backward link,
  and $\ens E_\out$ consists of a unique outlink, starting from the last node of the chain.
\end{theo}
  \begin{figure}[!hbt]
    \begin{center}
    \includegraphics{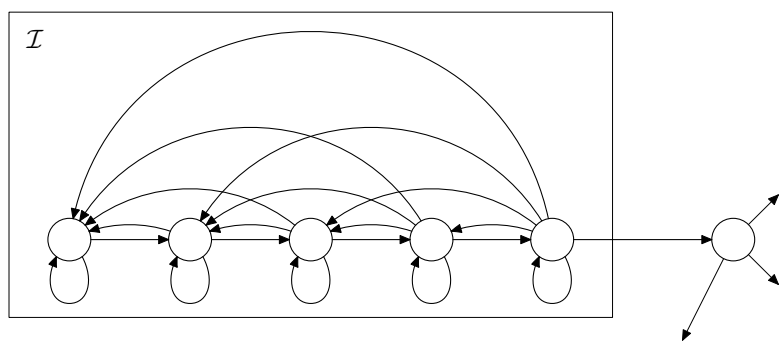}
    \caption{
      Every optimal linkage strategy for a set $\intr$ of five pages
      must have this structure.}
    \label{fig:optimal-link-structure}
    \end{center}
  \end{figure}

\medskip
This paper is organized as follows.  In the following preliminary section, we recall some graph concepts as well as the definition of the PageRank, and we introduce some notations. In Section~\ref{sec:PR-set}, we develop tools for analysing the PageRank of a set of pages~$\intr$. Then we come to the main part of this paper: in Section~\ref{sec:optimal-linkage-strategy} we provide the optimal linkage strategy for a set of nodes. In Section~\ref{sec:extensions-variants}, we give some extensions and variants of the main theorems.
We end this paper with some concluding remarks.

\section{Graphs and PageRank}\label{sec:preliminaries}

Let $\mathcal G=(\ens N,\ens E)$ be a directed graph representing the Web. The webpages are represented by the set of nodes $\ens N=\set{1,\dots,n}$ and the hyperlinks are represented by the set of directed links $\ens E\subseteq\ens N\times \ens N$. That means that $(i,j)\in\ens E$ if and only if there exists a hyperlink linking page~$i$ to page~$j$.

Let us first briefly recall some usual concepts about directed graphs (see for instance~\cite{BP94}).
A link $(i,j)$ is said to be an \mdef{outlink} for node $i$ and an \mdef{inlink} for node $j$.
If $(i,j)\in\ens E$, node $i$ is called a \mdef{parent} of node $j$. By \[j\ch{i},\] we mean that $j$ belongs to the set of \mdef{children} of $i$, that is $j\in\set{k\in\ens N\colon(i,k)\in\ens E}$.
The \mdef{outdegree} $\dg{i}$ of a node $i$ is its number of children, that is
\[ \dg{i}=\abs{\set{j\in\ens N\colon (i,j)\in\ens E}}. \]
A \mdef{path} from $i_0$ to $i_s$ is a sequence of nodes $\langle i_0,i_1,\dots,i_s\rangle$ such that $(i_k,i_{k+1})\in\ens E$ for every $k=0,1,\dots,s-1$. A node $i$ \mdef{has an access to a node} $j$ if there exists a path from $i$ to $j$. In this paper, we will also say that a node $i$ \mdef{has an access to a set} $\ens J$ if $i$ has an access to at least one node $j\in\ens J$.  The graph $\ens G$ is \emph{strongly connected} if every node \new{of} $\ens N$ has an access to every other node \new{of} $\ens N$.
\new{A set of nodes $\ens F\subseteq\ens N$ is a \mdef{final class} of the graph $\mathcal G=(\ens N,\ens E)$ if the subgraph $(\ens F,\ens E_{\ens F})$ is strongly connected and moreover $\ens E_{\mathrm{out}(\ens F)}=\emptyset$ (i.e.\ nodes of $\ens F$ do not have an access to $\ens N\setminus\ens F$).}

\medskip
Let us now briefly introduce the PageRank score (see~\cite{BGS05,BP98,LM04,LM06book,PBMW98} for background). Without loss of generality (please refer to the book of Langville and Meyer~\cite{LM06book} or the survey of Bianchini et al.~\cite{BGS05} for details), we can make the assumption that \emph{each node has at least one outlink}, i.e.\ $\dg{i}\neq0 \text{ for every } i\in\ens N$.
Therefore the $n\times n$ stochastic matrix $P=[P_{ij}]_{i,j\in\ens N}$ given by
\[
  P_{ij}=\begin{cases}
    {\dg{i}}^{-1} &\text{if }(i,j)\in\ens E,\\
    0 &\text{otherwise},
  \end{cases}
\]
is well defined and is a scaling of the adjacency matrix of $\mathcal G$.
Let also $0<c<1$ be a \mdef{damping factor} and $\vec z$ be a positive stochastic \mdef{personalization vector}, i.e.\ \new{$\vec z_i>0$ for all $i=1,\dots,n$ and} $\vec z^T\vun=1$,
where $\vun$ denotes the vector of all ones.
The \mdef{Google matrix} is then defined as
\[ G=cP+(1-c)\vun\vec z^T. \]
Since $\vec z>0$ and $c<1$, this stochastic matrix is positive, \new{i.e.\ $G_{ij}>0$ for all $i,j$}.  The \mdef{PageRank vector} $\prv$ is then defined as the unique invariant measure of the matrix $G$, \new{that is} the unique left Perron vector of $G$,
\begin{equation}\label{eq:PR-definition}
\begin{split}
  \prv^T &=\prv^TG,\\
  \prv^T\vun&=1.
\end{split}
\end{equation}
The \mdef{PageRank of a node} $i$ is the $i^\text{th}$ entry $\prv_i=\prv^T\vec e_i$ of the PageRank vector.

The PageRank vector is usually interpreted as the stationary distribution of the following Markov chain (see for instance~\cite{LM06book})\new{: a} random surfer moves on the webgraph, using hyperlinks between pages with a probability $c$ and \mdef{zapping} to some new page according to the personalization vector with a probability $(1-c)$. The Google matrix $G$ is the probability transition matrix of this random walk.
In this stochastic interpretation, the PageRank of a node is equal to the inverse of its mean return time, that is \new{$\prv_i^{-1}$} is the mean number of steps a random surfer starting in node~$i$ will take for coming back to~$i$  (see~\cite{CM00,KS60}).

\section{PageRank of a website}\label{sec:PR-set}

We are interested in characterizing the \mdef{PageRank of a set}~$\intr$. We define this as the sum
\[\prv^T\vec e_\intr=\sum_{i\in\intr}\prv_i,\]
where $\vec e_\intr$ denotes the vector with a $1$ in the entries of $\intr$ and $0$ elsewhere. Note that the PageRank of a set corresponds to the notion of energy of a community in~\cite{BGS05}.

Let $\intr\subseteq\ens N$ be a subset of the nodes of the graph.
The PageRank of $\intr$ can be expressed as $\prv^T\vec e_\intr=(1-c)\vec z^T(I-cP)^{-1}\vec e_\intr$ from PageRank equations~\eqref{eq:PR-definition}.
Let us then define the vector
\begin{equation}\label{eq:v-definition} \vec v = (I-cP)^{-1}\vec e_\intr.\end{equation}
With this, we have the following expression for the PageRank of the set~$\intr$:
\begin{equation}\label{eq:PR-zv} \prv^T\vec e_\intr=(1-c)\vec z^T\vec v.\end{equation}

The vector $\vec v$ will play a crucial role throughout this paper.  In this section, we will first present a probabilistic interpretation for this vector and prove some of its properties. We will then show how it can be used in order to analyze the influence of some page $i\in\intr$ on the PageRank of the set $\intr$. We will end this section by briefly introducing the concept of \new{basic} absorbing graph, which will be useful in order to analyze optimal linkage strategies under some assumptions.

\subsection{Mean number of visits before zapping}

Let us first see how the entries of the vector $\vec v=(I-cP)^{-1}\vec e_\intr$ can be interpreted. Let us consider a random surfer on the webgraph $\ens G$ that, as described in Section~\ref{sec:preliminaries}, follows the hyperlinks of the webgraph with a probability $c$. But, instead of zapping to some page of $\ens G$ with a probability $(1-c)$, he \emph{stops} his walk with probability $(1-c)$ at each step of time. This is equivalent to consider a random walk on
the extended graph $\mathcal{G}_e=(\ens N\cup\set{n+1},\ens E\cup\set{(i,n+1)\colon i\in\ens N})$ with a transition probability matrix
\[ P_e = \begin{pmatrix}
     cP&(1-c)\vun\\0&1
  \end{pmatrix}.
\]
At each step of time, with probability $1-c$, the random surfer can \emph{disappear} from the original graph, that is he can reach the absorbing node $n+1$.

The nonnegative matrix $(I-cP)^{-1}$ is commonly called the fundamental matrix of the absorbing Markov chain defined by $P_e$ (see for instance~\cite{KS60,Sen81}). In the extended graph~$\mathcal{G}_e$, the entry $[(I-cP)^{-1}]_{ij}$ is the expected number of visits to node~$j$ before reaching the absorbing node~$n+1$ when starting from node~$i$.
From the point of view of the standard random surfer described in Section~\ref{sec:preliminaries}, the entry $[(I-cP)^{-1}]_{ij}$ is the expected number of visits to node~$j$ before zapping for the first time when starting from node~$i$.

Therefore, the vector $\vec v$ defined in equation~\eqref{eq:v-definition} has the following probabilistic interpretation. The entry~$\vec v_i$ is the \emph{expected number of visits to the set~$\ens{I}$ before zapping} for the first time when the random surfer starts his walk in node~$i$.

\medskip
Now, let us first prove some simple properties about this vector.
\begin{lemma}\label{lem:minmaxv}
  Let $\vec v\in\R^n_{\ge0}$ be defined by $\vec v=cP\vec v+\vec e_\intr$.
  Then,
  \begin{itemize}
    \item[$(a)$] $\max_{i\notin\ens{I}}\vec v_i \le c\,\max_{i\in\ens{I}} \vec v_i$,
    \item[$(b)$] $\vec v_i\le 1+c\,\vec v_i$ for all $i\in\ens N$;
          with equality if and only if
          the node $i$ does not have an access to $\ext$,
    \item[$(c)$] $\vec v_i\ge\min_{j\ch{i}}\vec v_j$ for all $i\in\intr$;
          with equality if and only if
          the node $i$ does not have an access to $\ext$;
  \end{itemize}
\end{lemma}
\begin{proof}
  \begin{itemize}
    \item[$(a)$]  Since $c<1$, for all $i\notin\intr$,
  \[
    \max_{i\notin\intr}\vec v_i
    =\max_{i\notin\intr}\bigg(c\sum_{j\ch{i}}\frac{\vec v_j}{\dg{i}}\bigg)
    \le c\max_j\vec v_j.
  \]
  Since $c<1$, it then follows that $\max_{j} \vec v_j = \max_{i\in\intr} \vec v_i$.
  \item[$(b)$] The inequality $\vec v_i\le \frac{1}{1-c}$ follows directly from
  \[
    \max_i\vec v_i
    \le\max_i\bigg(1+c\sum_{j\ch{i}}\frac{\vec v_j}{\dg{i}}\bigg)
    \le 1+c\max_j\vec v_j.
  \]
  From~$(a)$ it then also follows that $\vec v_i\le \frac{c}{1-c}$ for all $i\notin\intr$.
  Now, let $i\in\ens N$ such that $\vec v_i= \frac{1}{1-c}$.
  Then $i\in\intr$.
  Moreover,
  \[
     1+c\,\vec v_i=\vec v_i=1+c\sum_{j\ch{i}}\frac{\vec v_j}{\dg{i}},
  \]
  that is $\vec v_j= \frac{1}{1-c}$ for every $j\ch{i}$.
  Hence node $j$ must also belong to~$\intr$. By induction,
  every node $k$ such that $i$ has an access to $k$
  must belong to~$\intr$.
  \item[$(c)$] Let $i\in\intr$. Then, by~$(b)$
  \[
     1+c\,\vec v_i\ge\vec v_i=1+c\sum_{j\ch{i}}\frac{\vec v_j}{\dg{i}}
     \ge 1+c\min_{j\ch{i}}\new{\vec v_j},
  \]
  so $\vec v_i\ge\min_{j\ch{i}}\vec v_j$ for all $i\in\intr$.
  If $\vec v_i=\min_{j\ch{i}}\vec v_j$ then also $1+c\,\vec v_i=\vec v_i$ and hence, by~$(b)$,
  the node $i$ does not have an access to $\ext$.
  \qedhere
  \end{itemize}
\end{proof}

Let us denote the set of nodes \new{of} $\ext$ which on average give the most visits \new{to} $\intr$ before zapping by
  \[ \maxV =\argmax_{j\in\ext}\vec v_j .\]
Then the following lemma is quite intuitive. It says that, among the nodes \new{of} $\ext$, those which provide the higher mean number of visits \new{to $\intr$ are parents of $\intr$, i.e.\ parents of some node of $\intr$}.
\begin{lemma}[Parents of $\intr$]\label{lem:vmax-parents}
If $\ens E_\inn\neq\emptyset$, then
\[ \maxV \subseteq
    \set{j\in\ext\colon 
    \text{ there exists } \ell\in\intr
    \text{ such that } (j,\ell)\in\ens E_\inn}.
\]
If $\ens E_\inn=\emptyset$, then $\vec v_j=0$ for every $j\in\ext$.
\end{lemma}
\begin{proof}
  Suppose first that $\ens E_\inn\neq\emptyset$.
  Let $k\in\maxV$
  with $\vec v=(I-cP)^{-1}\vec e_\intr$.
  If we supposed that there does not exist $\ell\in\intr$
  such that $(k,\ell)\in\ens E_\inn$, then we would have, since $\vec v_k>0$,
  \[
    \vec v_k=c\sum_{j\ch{k}}\frac{\vec v_j}{\dg{k}}
    \le c\max_{j\notin\intr}\vec v_j=c\vec v_k<\vec v_k,
  \]
  which is a contradiction.
  Now, if $\ens E_\inn=\emptyset$, then \new{there is no access to~$\intr$ from~$\ext$, so}
  clearly $\vec v_j=0$ for every $j\in\ext$.
\end{proof}
Lemma~\ref{lem:vmax-parents} shows that the nodes $j\in\ext$ which provide the higher value of $\vec v_j$ must belong to the set of \new{parents of $\intr$}. The converse is not true, as we will see in the following example: some parents of $\intr$ can provide a lower mean number of visits \new{to} $\intr$ that other nodes which are not parents of $\intr$. In other word, Lemma~\ref{lem:vmax-parents} gives a necessary but not sufficient condition in order to maximize the entry $\vec v_j$ for some $j\in\ext$.
\begin{example}\label{ex:parents-not-sufficient}
  Let us see on an example that having $(j,i)\in\ens E_\inn$ for some $i\in\intr$ is
  not sufficient to have $j\in\maxV $.
  Consider the graph in Figure~\ref{fig:parents-not-sufficient}. Let $\intr=\set{1}$
  and take a damping factor $c=0.85$.
  For $\vec v = (I-cP)^{-1}\vec e_1$, we have
  \[
    \vec v_2=\vec v_3=\vec v_4 =4.359> \vec v_5= 3.521
    >\vec v_6=3.492>\vec v_{7}>\cdots>\vec v_{11},
  \]
  so $\maxV =\set{2,3,4}$.
  As ensured by Lemma~\ref{lem:vmax-parents}, every node \new{of} the set
  $\maxV $ is a parent of node $1$.
  But here, $\maxV $ does not contain all parents of \new{node}~$1$. Indeed, the node~$6\notin\maxV $
  while it is a parent of~$1$ and is moreover its parent with the \new{lowest} outdegree.
  Moreover, we see in this example that \new{node} $5$, which is a not a parent of node $1$
  but a parent of \new{node} $6$, gives a higher value of
  the expected number of visits \new{to} $\intr$ before zapping, than node $6$, parent of $1$.
  Let us try to get some intuition about that.
  When starting from node $6$, a random surfer has probability one half
  to reach node $1$ in only one step.
  But he has also a probability one half to move to node $11$ and to be send
  far away from node $1$.
  On the other side, when starting from node $5$, the random surfer can not reach node $1$ in
  only one step. But with probability $3/4$ he will reach one of the nodes $2$, $3$ or $4$ in one step.
  And from these nodes, the websurfer \new{stays} very near to node $1$ and can not be \new{sent}
  far away from it.
  \begin{figure}[tb]
    \begin{center} \includegraphics{figures.1}
    \caption{The node $6\notin\maxV$
    \new{and yet} it is a parent of $\intr=\set{1}$
    (see Example~\ref{ex:parents-not-sufficient}).}
    \label{fig:parents-not-sufficient}
    \end{center}
  \end{figure}
\end{example}

In the next lemma, we show that from some node $i\in\intr$ which \new{has} an access to $\ext$, there always exists what we call a \mdef{decreasing path} to $\ext$. That is, we can find a path such that the mean number of visits \new{to} $\intr$ is higher when starting from some node of the path than when starting from the successor of this node in the path.
\begin{lemma}[Decreasing paths to $\ext$]\label{lem:decreasing-path}
  For every $i_0\in\intr$ which has an access to $\overline{\ens{I}}$,
  there exists a path $\langle i_0,i_1,\dots,i_s\rangle$ with $i_1,\dots,i_{s-1}\in\intr$ and
  $i_s\in\ext$
  such that \[  \vec v_{i_0}>\vec v_{i_1}>...>\vec v_{i_s}.  \]
\end{lemma}
\begin{proof}
  Let us simply construct a decreasing path recursively by
  \[ i_{k+1}\in\argmin_{j\ch{i_k}}\vec v_j,\]
  as long as $i_k\in\intr$.
  If $i_k$ has an access to~$\ext$, then $\vec v_{i_{k+1}}<\vec v_{i_k}<\frac{1}{1-c}$
  by Lemma~\ref{lem:minmaxv}$(b)$ and $(c)$, so the node~$i_{k+1}$ has also an access to~$\ext$.
  By \new{assumption}, $i_0$ has an access to~$\ext$. Moreover, the set $\intr$ has a finite
  number of elements, so there must exist \new{an} $s$ such that $i_s\in\ext$.
\end{proof}

\subsection{Influence of the outlinks of a node}

We will now see how a modification of the outlinks of some node $i\in\ens N$ can change the PageRank of a subset of nodes $\intr\subseteq\ens N$. So we will compare two graphs on $\ens N$ defined by their set of links, $\ens E$ and $\widetilde{\ens E}$ respectively.

Every \new{item} corresponding to the graph defined by the set of links $\widetilde{\ens E}$ will be written with a tilde symbol. So $\widetilde{P}$ denotes its scaled adjacency matrix, $\prvt$ the corresponding PageRank vector, $\dgt{i}=\abs{\set{j\colon (i,j)\in\widetilde{\ens E}}}$ the outdegree of some node~$i$ in this graph, $\widetilde{\vec v}=(I-c\widetilde P)^{-1}\vec e_\intr$ and $\maxVt=\argmax_{j\in\ext}\widetilde{\vec v}_j$. Finally, by $j\cht{i}$ we mean $j\in\set{k\colon(i,k)\in\widetilde{\ens E}}$.

So, let us consider two graphs defined respectively by their set of links $\ens E$ and $\widetilde{\ens E}$. Suppose that they differ only in the links starting from some given node $i$, that is $\set{j\colon (k,j)\in\ens E}=\set{j\colon (k,j)\in\widetilde{\ens E}}$ for all $k\neq i$. Then their scaled adjacency matrices $P$ and $\widetilde P$ are linked by a rank one correction. 
Let us then define the vector
  \begin{equation*}
    \vec\delta=\sum_{j\cht{i}}\frac{\vec e_j}{\dgt{i}}
      -\sum_{j\ch{i}}\frac{\vec e_j}{\dg{i}},
  \end{equation*}
which gives the correction to apply to the line $i$ of the matrix $P$ in order to get $\widetilde P$.

Now let us first express the difference between the PageRank of $\intr$ for two configurations differing only in the links starting from some node~$i$. Note \new{that} in the following lemma the personalization vector $\vec z$ does not appear explicitly in the expression of $\prvt$.
\begin{lemma}\label{lem:pt-pi}
  Let two graphs defined respectively by $\ens E$ and $\widetilde{\ens E}$
  and let $i\in\ens N$ such that for all $k\neq i$,
  $\set{j\colon (k,j)\in\ens E}=\set{j\colon (k,j)\in\widetilde{\ens E}}$.
  Then
  \[
    \prvt^T \vec e_\intr 
     =\prv^T\vec e_\intr + c\,\prv_i\,
        \frac{\vec\delta^T \vec v }{1-c\,\vec\delta^T(I-cP)^{-1}\vec e_i}.
  \]
\end{lemma}
\begin{proof}
  Clearly, the scaled adjacency matrices are linked by $\widetilde P = P + \vec e_i\, \vec\delta^T$.
  Since $c<1$, the matrix $(I-cP)^{-1}$ exists and the PageRank vectors can be expressed as
  \begin{align*}
   \prv^T&=(1-c)\vec z^T(I-cP)^{-1},\\
   \prvt^T&=(1-c)\vec z^T(I-c\,(P+\vec e_i\vec\delta^T))^{-1}.
  \end{align*}
  Applying the Sherman--Morrison formula \new{to} $((I-cP)-c\vec e_i\vec\delta^T)^{-1}$, we get
  \[
    \prvt^T =(1-c)\vec z^T(I-cP)^{-1}
      + (1-c)\vec z^T(I-cP)^{-1}\vec e_i\,
        \frac{c\vec\delta^T (I-cP)^{-1}}{1-c\vec\delta^T(I-cP)^{-1}\vec e_i},
  \]
  and the result follows immediately.
\end{proof}

Let us now give an equivalent condition in order to increase the PageRank of $\intr$ by changing outlinks of some node $i$. The PageRank of~$\intr$ increases \new{essentially} when the new set of links favors nodes giving a higher mean number of visits \new{to} $\intr$ before zapping.
\begin{theorem}[PageRank and mean number of visits before zapping]\label{thm:p>p-d>0}
  Let two graphs defined respectively by $\ens E$ and $\widetilde{\ens E}$
  and let $i\in\ens N$ such that for all $k\neq i$,
  $\set{j\colon (k,j)\in\ens E}=\set{j\colon (k,j)\in\widetilde{\ens E}}$.
  Then
  \[ \prvt^T \vec e_\intr >\prv^T\vec e_\intr \quad \text{ if and only if } \quad
     \vec\delta^T\vec v>0 \]
  and $\prvt^T \vec e_\intr =\prv^T\vec e_\intr$ if and only if
  $\vec\delta^T\vec v=0$.
\end{theorem}
\begin{proof}
  Let us first show that $\vec\delta^T(I-cP)^{-1}\vec e_i\le1$ is always verified.
  Let $\vec u =(I-cP)^{-1}\vec e_i$.  Then $\vec u-cP\vec u=\vec e_i$ and,
  by Lemma~\ref{lem:minmaxv}$(a)$, $\vec u_j\le\vec u_i$ for all $j$.
  So
  \[
    \vec\delta^T\vec u
      = \sum_{j\cht{i}}\frac{\vec u_j}{\dgt{i}}
      -\sum_{j\ch{i}}\frac{\vec u_j}{\dg{i}}
      \le \vec u_i - \sum_{j\ch{i}}\frac{\vec u_j}{\dg{i}}
      \le \vec u_i - c\sum_{j\ch{i}}\frac{\vec u_j}{\dg{i}} =1.
   \]
  Now, since $c<1$ and $\prv>0$, the conclusion follows by Lemma~\ref{lem:pt-pi}.
\end{proof}

The following Proposition~\ref{prop:add-link} shows how to \new{add a new} link $(i,j)$ starting from a given node $i$ in order to increase the PageRank of the set $\intr$. The PageRank of $\intr$ increases as soon as a node $i\in\intr$ adds a link to a node \new{$j$ with a larger or equal expected number of visits to $\intr$ before zapping}.
\begin{proposition}[Adding a link]\label{prop:add-link}
  Let $i\in\intr$ and let
  $j\in\ens N$ \new{be} such that $(i,j)\notin\ens E$ and $\vec v_i\le \vec v_j$.
  Let $\widetilde{\ens E}=\ens E\cup\set{(i,j)}$.
  Then \[ \prvt^T\vec e_{\intr} \ge \prv^T\vec e_{\intr} \]
  with equality if and only \new{if} the node $i$ does not have an access to $\ext$.
\end{proposition}
\begin{proof}
  Let $i\in\intr$ and let $j\in\ens N$ \new{be}
  such that $(i,j)\notin\ens E$ and $\vec v_i\le \vec v_j$.
  Then
  \[
    1+c\sum_{k\ch{i}}\frac{\vec v_k}{\dg{i}}=\vec v_i\le1+c\vec v_i\le 1+c\vec v_j,
   \]
  with equality if and only if $i$ does not have an access to $\ext$
  by Lemma~\ref{lem:minmaxv}$(b)$.
  Let $\widetilde{\ens E}=\ens E\cup\set{(i,j)}$. Then
  \[
    \vec\delta^T\vec v
      = \frac{1}{\dg{i}+1}
        \bigg( \vec v_j-\sum_{k\ch{i}}\frac{\vec v_k}{\dg{i}}\bigg)
      \ge0,
  \]
  with equality if and only \new{if} $i$ does not have an access to $\ext$.
  \new{The conclusion follows from} Theorem~\ref{thm:p>p-d>0}.
\end{proof}

Now let us see how to \new{remove} a link $(i,j)$ starting from a given node $i$ in order to increase the PageRank of the set $\intr$. If a node $i\in\ens N$ removes a link to its worst child from the point of view of the expected number of visits to $\intr$ before zapping, then the PageRank of $\intr$ increases.
\begin{proposition}[Removing a link]\label{prop:remove-link}
  Let $i\in\ens N$ and let
  $j\in\argmin_{k\ch{i}} \vec v_k$.
  Let $\widetilde{\ens E}=\ens E\setminus\set{(i,j)}$.
  Then \[ \prvt^T\vec e_{\intr} \ge \prv^T\vec e_{\intr} \]
  with equality if and only \new{if} $\vec v_k=\vec v_j$ for every $k$ such that $(i,k)\in\ens E$.
\end{proposition}
\begin{proof}
  Let $i\in\ens N$ and let
  $j\in\argmin_{k\ch{i}} \vec v_k$. Let $\widetilde{\ens E}=\ens E\setminus\set{(i,j)}$.
  Then
  \[
    \vec\delta^T\vec v
      = \sum_{k\ch{i}}\frac{\vec v_k-\vec v_j}{\dg{i}(\dg{i}-1)}
      \ge0,
  \]
  with equality if and only if $\vec v_k=\vec v_j$ for all $k\ch{i}$.
  The conclusion follows by Theorem~\ref{thm:p>p-d>0}.
\end{proof}

\new{In order to increase the PageRank of $\ens I$ with a new link $(i,j)$, Proposition~\ref{prop:add-link} only requires that $\vec v_j\le\vec v_i$.
On the other side, Proposition~\ref{prop:remove-link} requires that $\vec v_j=\min_{k\ch{i}}\vec v_k$
in order to increase the PageRank of $\ens I$ by deleting link $(i,j)$.
One could wonder whether or not this condition could be weakened to $\vec v_j<\vec v_i$,
so as to have symmetric conditions for the addition or deletion of links.
In fact, this can not be done as shown in the following example.}
\begin{example}\label{ex:v<v-not-sufficient}
  Let us see by an example that the condition $j\in\argmin_{k\ch{i}} \vec v_k$
  in Proposition~\ref{prop:remove-link} can not be weakened \new{to} $\vec v_j<\vec v_i$.
  Consider the graph in Figure~\ref{fig:v<v-not-sufficient} and take a damping factor $c=0.85$.
  Let $\intr=\set{1,2,3}$. We have
  \[
  \vec v_1=2.63>\vec v_2=2.303>\vec v_3 =1.533.
  \]
  As ensured by Proposition~\ref{prop:remove-link},
  if we remove the link $(1,3)$, the PageRank of $\intr$ increases
  (e.g.\ from $0.199$ to $0.22$ with a uniform personalization vector $\vec z=\frac{1}{n}\vun$),
  since $3\in\argmin_{k\ch{1}}\vec v_k$.
  But, if we remove instead the link $(1,2)$, the PageRank of $\intr$ decreases
  (from $0.199$ to $0.179$ with $\vec z$ uniform) even if $\vec v_2<\vec v_1$.
  \begin{figure}[htb]
    \begin{center} \includegraphics{figures.2}
    \caption{For $\intr=\set{1,2,3}$,
    removing link~$(1,2)$ gives $\prvt^T\vec e_\intr<\prv^T\vec e_\intr$,
    even if $\vec v_1>\vec v_2$ (see Example~\ref{ex:v<v-not-sufficient}).}
    \label{fig:v<v-not-sufficient}
    \end{center}
  \end{figure}
\end{example}

\begin{remark}
  Let us note that, if the node $i$ does not have an access to the set $\ext$,
  then for every \new{ \emph{deletion} of a} link starting from $i$,
  the PageRank of $\intr$ will not be modified.
  Indeed, in this case $\vec\delta^T\vec v=0$ since by Lemma~\ref{lem:minmaxv}$(b)$,
  $\vec v_j=\frac{1}{1-c}$ for every $j\ch{i}$.
\end{remark}

\subsection{Basic absorbing graph}

Now, let us introduce briefly the notion of basic absorbing graph (see Chapter~III about absorbing Markov chains in Kemeny and Snell's book~\cite{KS60}).

For a given graph $(\ens N,\ens E)$ and a specified subset of nodes $\intr\subseteq\ens N$, the \mdef{basic absorbing graph} is the graph $(\ens N,\ens E^0)$ defined by $\ens E_\out^0=\emptyset$, $\ens E_\intr^0=\set{(i,i)\colon i\in\intr}$, $\ens E_\inn^0=\ens E_\inn$ and $\ens E_\ext^0=\ens E_\ext$. In other words, the basic absorbing graph $(\ens N,\ens E^0)$ is a graph constructed from $(\ens N,\ens E)$, keeping the same sets of external inlinks and external links \new{$\ens E_\inn,\ens E_\ext$}, removing the external outlinks \new{$\ens E_\out$} and changing the internal link structure \new{$\ens E_\intr$} in order to have only self-links for nodes of~$\intr$.

Like in the previous subsection, every \new{item} corresponding to the basic absorbing graph will \new{have} a zero symbol. For instance, we will write $\prv_0$ for the PageRank vector corresponding to the basic absorbing graph and $\maxV_0=\argmax_{j\in\ext}[(I-cP_0)^{-1}\vec e_\intr]_j$.

\begin{proposition}[PageRank for a basic absorbing graph]\label{prop:out0-best}
  Let a graph defined by a set of links $\ens E$ and let $\intr\subseteq\ens N$.
  Then
  \[ \prv^T\vec e_\intr\le\prv_0^T\vec e_\intr, \]
  with equality if and only if ${\ens E}_\out=\emptyset$.
\end{proposition}
\begin{proof}
  Up to a permutation of the indices, equation~\eqref{eq:v-definition} can be written as
  \[
  \begin{pmatrix}I-cP_{\intr}&-cP_{\out}\\-cP_{\inn}&I-cP_{\ext}\end{pmatrix}
  \begin{pmatrix}\vec v_\intr\\\vec v_\ext\end{pmatrix}
  =\begin{pmatrix}\vun\\0\end{pmatrix},
  \]
  so we get
  \begin{equation}\label{eq:vintr-vext}
  \vec v = \begin{pmatrix}\vec v_\intr\\c (I-cP_{\ext})^{-1}P_\inn\vec v_\intr\end{pmatrix}.
  \end{equation}
  By Lemma~\ref{lem:minmaxv}(b) and since $(I-cP_{\ext})^{-1}$ is a nonnegative matrix
  (see for instance the chapter on $M$-matrices in Berman and Plemmons's book~\cite{BP94}),
  we then have
  \[ \vec v
    \le \begin{pmatrix}\frac{1}{1-c}\,\vun\\
      \frac{c}{1-c}\, (I-cP_{\ext})^{-1}P_\inn\vun\end{pmatrix}=\vec v_0, \]
  with equality if and only if no node of~$\intr$ has an access to~$\ext$, that is 
  $\ens E_\out=\emptyset$. The conclusion now follows from equation~\eqref{eq:PR-zv} and
  $\vec z>0$.
\end{proof}

Let us finally prove a nice property of the set $\maxV $ when $\intr=\set{i}$ is a singleton: it is independent of the outlinks of $i$. In particular, it can be found from the basic absorbing graph.
\begin{lemma}\label{lem:V=V0-single}
  Let a graph defined by a set of links $\ens E$ and let $\intr=\set{i}$
  Then there exists \new{an} $\alpha\neq0$ such that
  $(I-cP)^{-1}\vec e_i=\alpha(I-cP_0)^{-1}\vec e_i$.
  As a consequence,
  \[
    \maxV =\maxV_0 .
  \]
\end{lemma}
\begin{proof}
  Let $\intr=\set{i}$.
  Since $\vec v_\intr=\vec v_i$ is a scalar, it follows from equation~\eqref{eq:vintr-vext} that
  the direction of the vector $\vec v$ does not depend
  on $\ens E_\intr$ and $\ens E_\out$ but only on $\ens E_\inn$ and $\ens E_\ext$.
\end{proof}


\section{Optimal linkage strategy for a website}\label{sec:optimal-linkage-strategy}

In this section, we consider a set of nodes $\intr$. For this set, we want to choose the sets of internal links $\ens E_{\intr}\subseteq\intr\times \intr$ and external outlinks $\ens E_{\out}\subseteq\intr\times \ext$ in order to maximize the PageRank score of $\intr$, that is $\prv^T\vec e_\intr$.

Let us first discuss about the constraints on $\ens E$ we will consider. If we do not impose any condition on $\ens E$, the problem of maximizing $\prv^T\vec e_\intr$ is quite trivial.  As shown by Proposition~\ref{prop:out0-best}, you should take in this case $\ens E_\out=\emptyset$ and $\ens E_\intr$ an arbitrary subset of $\intr\times \intr$ such that each node has at least one outlink. You just try to \new{lure the random walker to} your pages, not allowing him to leave $\intr$ except by zapping according to the preference vector.
Therefore, it seems sensible to impose that \emph{$\ens E_\out$ must be nonempty}.

\new{Now, let us show that, in order to avoid trivial solutions to our maximization problem, it is not enough to assume that $\ens E_\out$ must be nonempty. Indeed, with this single constraint, in order to lose as few as possible visits from the random walker, you should take a unique leaking node $k\in\intr$ (i.e.\ $\ens E_\out=\set{(k,\ell)}$ for some $\ell\in\ext$) and isolate it from the rest of the set $\intr$ (i.e.\ $\set{i\in\intr\colon(i,k)\in\ens E_\intr}=\emptyset$).}

Moreover, it seems reasonable to imagine that Google penalizes (or at least \new{tries} to penalize) such behavior in the context of spam alliances~\cite{GGM05}.

\medskip
All this discussion leads us to make the following assumption.
\begin{assumption}[\new{Accessibility}]\label{ass:access-to-ext}
  Every node \new{of} $\intr$ has an access to at least one node \new{of} $\ext$.
\end{assumption}

\medskip
Let us now explain the basic ideas we will use in order to determine an optimal linkage strategy for a set of webpages $\intr$.
We determine some forbidden patterns for an optimal linkage strategy and deduce the only possible structure an optimal strategy can have. In other words, we assume that we have a configuration which gives an optimal PageRank $\prv^T\vec e_\intr$. Then we prove that if some particular pattern appeared in this optimal structure, then we could construct another graph for which the PageRank $\prvt^T\vec e_\intr$ is strictly higher than $\prv^T\vec e_\intr$.

\medskip
We will firstly determine the shape of an optimal external outlink structure $\ens E_\out$, when the internal link structure $\ens E_\intr$ is given, in Theorem~\ref{thm:out-opt}. Then, given the external outlink structure $\ens E_\out$ we will determine the possible optimal internal link structure $\ens E_\intr$ in Theorem~\ref{thm:intr-opt}. Finally, we will put both results together in Theorem~\ref{thm:website-opt} in order to get the general shape of an optimal linkage strategy for a set $\intr$ when $\ens E_\inn$ and $\ens E_\ext$ are given.

Proofs of this section will be illustrated by several figures \new{for which we take the following drawing convention}.
\begin{convention}
\new{When nodes are \new{drawn} from left to right on the same horizontal line, they are arranged by decreasing value of~$\vec v_j$. Links are represented by continuous arrows and paths by dashed arrows.}
\end{convention}

\medskip
The first result of this section concerns the optimal \emph{outlink} structure $\ens E_\out$ for the set $\intr$, while its internal structure $\ens E_\intr$ is given. An example of optimal outlink structure is given \new{after the theorem}.

\begin{theorem}[Optimal outlink structure]\label{thm:out-opt}
  Let $\ens E_\intr$, $\ens E_\inn$ and $\ens E_\ext$ be given.
  Let $\ens F_1, \dots,\ens F_r$ be the final classes \new{of the subgraph}
  $(\intr,\ens E_\intr)$.
  Let $\ens E_\out$ such that the PageRank $\prv^T\vec e_\intr$ is maximal 
  under Assumption~A. Then $\ens E_\out$ has the following structure:
  \[  \ens E_\out = \ens E_{\mathrm{out}(\ens F_1)} \cup\cdots\cup\ens E_{\mathrm{out}(\ens F_r)}, \]
  where for every $s=1,\dots,r$,
  \[ \ens E_{\mathrm{out}(\ens F_s)}
    \subseteq\set{(i,j)\colon i\in\argmin_{k\in\ens F_s}\vec{v}_k \text{ and } j\in\maxV}. \]
  Moreover for every $s=1,\dots,r$, if $\ens E_{\ens F_s}\neq\emptyset$,
  then $\abs{\ens E_{\mathrm{out}(\ens F_s)}}=1$.
\end{theorem}
\begin{proof}
  Let $\ens E_\intr$, $\ens E_\inn$ and $\ens E_\ext$ be given.
  Suppose $\ens E_\out$ is such that $\prv^T\vec e_\intr$ is maximal under
  Assumption~A.

  \new{We will determine the possible leaking nodes of $\intr$ by
  analyzing three different cases.}

  \new{Firstly, let us consider some node $i\in\intr$ such that $i$ 
  does not have children in $\intr$, i.e.\
  $\set{k\in\intr\colon (i,k)\in\ens E_\intr}=\emptyset$.
  Then clearly we have $\set{i}=\ens F_s$ for some $s=1,\dots,r$, with
  $i\in\argmin_{k\in\ens F_s}\vec v_k$ and $\ens E_{\ens F_s}=\emptyset$.
  From Assumption~A, we have $\ens E_{\mathrm{out}(\ens F_s)}\neq\emptyset$, and
  from Theorem~\ref{thm:p>p-d>0} and the optimality assumption, we have
  $\ens E_{\mathrm{out}(\ens F_s)}\subseteq\set{(i,j)\colon j\in\maxV}$
  (see Figure~\ref{fig:out-opt-fig2}).}
  \begin{figure}[!htb]
  \centering
  {\includegraphics{figuresthms.2}}%
    \caption{If $\vec v_j<\vec v_\ell$,
    then $\prvt^T\vec e_\intr>\prv^T\vec e_\intr$ with
    $\widetilde{\ens E}_\out=\ens E_\out\cup\set{(i,\ell)}\setminus\set{(i,j)}$.}
    \label{fig:out-opt-fig2}
  \end{figure}

  Secondly, let us consider some $i\in\intr$ such that \new{$i$ has children in $\intr$, i.e.\
  $\set{k\in\intr\colon (i,k)\in\ens E_\intr}\neq\emptyset$ and}
  \[\new{ \vec v_i\le\min_{\substack{k\ch{i}\\k\in\intr}}\vec v_k}. \]
  Let $j\in\argmin_{k\ch{i}}\vec v_k$.
  Then $j\in\ext$ and $\vec v_j<\vec v_i$ by Lemma~\ref{lem:minmaxv}$(c)$.
  Suppose by contradiction that the node $i$ would keep an access to $\ext$ if we took
  $\widetilde{\ens E}_\out=\ens E_\out\setminus\set{(i,j)}$ instead of $\ens E_\out$.
  Then, by Proposition~\ref{prop:remove-link}, considering $\widetilde{\ens E}_\out$
  instead of $\ens E_\out$ would increase strictly the PageRank of $\intr$
  while Assumption~A remains satisfied
  (see Figure~\ref{fig:out-opt-fig1}).
  \begin{figure}[!htb]
  \centering
  {\includegraphics{figuresthms.1}}
   \caption{If $\vec v_j=\min_{k\ch{i}}\vec v_k$
    and $i$ has another access to $\ext$, then $\prvt^T\vec e_\intr>\prv^T\vec e_\intr$
    with $\widetilde{\ens E}_\out=\ens E_\out\setminus\set{(i,j)}$.}
    \label{fig:out-opt-fig1}
  \end{figure}
  This would contradict the
  optimality \new{assumption} for $\ens E_\out$.
  From this, we conclude that\vspace{-1ex}
  \begin{itemize}
    \item the node $i$ belongs to final class $\ens F_s$ of \new{the subgraph}
       $(\intr,\ens E_\intr)$ with $\ens E_{\ens F_s}\neq\emptyset$
       for some $s=1,\dots,r$;\vspace{-1ex}
    \item there does not exist another $\ell\in\ext$, $\ell\neq j$
       such that $(i,\ell)\in\ens E_\out$;\vspace{-1ex}
    \item there does not exist another $k$ in the same final class $\ens F_s$, $k\neq i$ such that
      such that $(k,\ell)\in\ens E_\out$ for some $\ell\in\ext$.\vspace{-1ex}
  \end{itemize}
  \new{Again, by Theorem~\ref{thm:p>p-d>0} and the optimality assumption,
  we have $j\in\ens V$ (see Figure~\ref{fig:out-opt-fig2}).}

  \new{Let us now notice that}
  \begin{equation}\label{eq:maxbarIv<minIv}
    \max_{k\in\ext} \vec{v}_k<\min_{k\in\intr}\vec v_k.
  \end{equation}
  \new{Indeed, with $i\in\argmin_{k\in\intr}\vec v_k$, we are in one of the two cases
  analyzed above for which we have seen that
  $\vec v_i>\vec v_j=\argmax_{k\in\ext} \vec{v}_k$.  }

  Finally, consider a node $i\in\intr$ that does not belong to any of the final classes
  of \new{the subgraph} $(\intr,\ens E_\intr)$.
  Suppose by contradiction that there exists $j\in\ext$ such that $(i,j)\in\ens E_\out$.
  Let $\ell\in\argmin_{k\ch{i}}\vec v_k$.
  Then it follows from inequality~\eqref{eq:maxbarIv<minIv} that $\ell\in\ext$.
  But the same argument as above shows that the link $(i,\ell)\in\ens E_\out$
  must be removed since $\ens E_\out$ is supposed to be optimal
  (see Figure~\ref{fig:out-opt-fig1} again).
  So, there does not exist $j\in\ext$ such that $(i,j)\in\ens E_\out$
  for a node $i\in\intr$ which does not belong to any of the final classes
  $\ens F_1,\dots,\ens F_r$.
\end{proof}

\begin{example}\label{ex:out-opt-ex}
  \new{Let us consider the graph given in Figure~\ref{fig:out-opt-ex}.
  The internal link structure $\ens E_\intr$, as well as $\ens E_\inn$ and $\ens E_\ext$ are given.
  The subgraph $(\intr,\ens E_\intr)$ has two final classes $\ens F_1$ and $\ens F_2$.
  With $c=0.85$ and $\vec z$ the uniform probability vector,
  this configuration has six optimal outlink structures
  (one of these solutions is represented by bold arrows in Figure~\ref{fig:out-opt-ex}).
  Each one can be written as
  $\ens E_\out=\ens E_{\mathrm{out}(\ens F_1)}\cup\ens E_{\mathrm{out}(\ens F_2)}$,
  with $\ens E_{\mathrm{out}(\ens F_1)}=\set{(4,6)}$ or
  $\ens E_{\mathrm{out}(\ens F_1)}=\set{(4,7)}$
  and $\emptyset\neq\ens E_{\mathrm{out}(\ens F_2)}\subseteq\set{(5,6),(5,7)}$.
  Indeed, since $\ens E_{\ens F_1}\neq\emptyset$, as stated by Theorem~\ref{thm:out-opt},
  the final class $\ens F_1$ has exactly one external outlink in every optimal outlink structure.
  On the other hand, the final class $\ens F_2$ may have several
  external outlinks, since it is composed of a unique node and moreover
  this node does not have a self-link. Note that $\ens V=\set{6,7}$
  in each of these six optimal configurations,
  but this set $\ens V$ can not be determined a priori since it depends on the
  chosen outlink structure.}
  \begin{figure}[!htb]
  \centering
  {\includegraphics{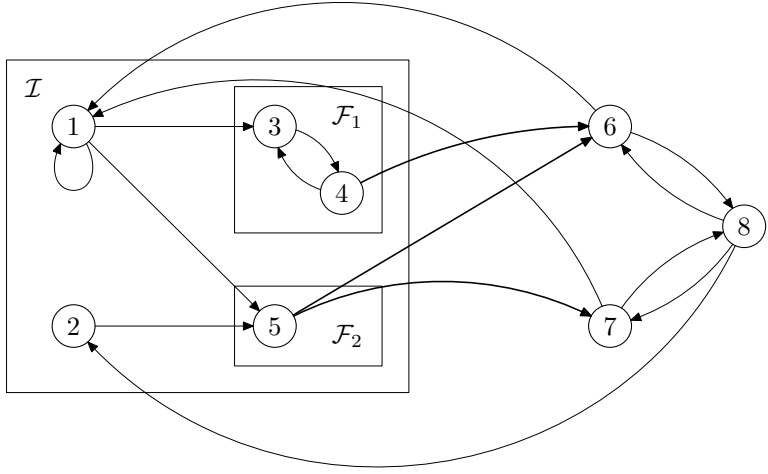}}%
    \caption{Bold arrows represent one of the six optimal \emph{outlink} structures
    for this configuration with two final classes (see Example~\ref{ex:out-opt-ex}).}
    \label{fig:out-opt-ex}
  \end{figure}
\end{example}

Now, let us determine the optimal \emph{internal} link structure $\ens E_\intr$ for the set $\intr$, while its outlink structure $\ens E_\out$ is given. Examples of optimal internal structure are given \new{after the proof of the theorem}.
\begin{theorem}[Optimal internal link structure]\label{thm:intr-opt}
  Let $\ens E_\out$, $\ens E_\inn$ and $\ens E_\ext$ be given.
  Let $\ens L=\set{i\in \intr\colon (i,j)\in\ens E_\out \text{ for some }j\in\ext}$
  be the set of leaking nodes \new{of} $\intr$ and
  let $n_{\ens L}=\abs{\ens L}$ \new{be} the number of leaking nodes.
  Let $\ens E_\intr$ such that the PageRank $\prv^T\vec e_\intr$ is maximal 
  under Assumption~A.
  Then there exists a permutation of the indices such that $\intr=\set{1,2,\dots,n_\intr}$,
  $\ens L=\set{n_\intr-n_{\ens L}+1,\dots,n_\intr}$,
  \[
    \vec v_1>\cdots>\vec v_{{n_\intr}-n_{\ens L}}
       >\vec v_{{n_\intr}-n_{\ens L}+1}\ge\cdots\ge\vec v_{{n_\intr}},
  \]
  and $\ens E_\intr$ has the following structure:
  \[\ens E_\intr^L\subseteq\ens E_\intr\subseteq\ens E_\intr^U,\]
  where
  \begin{align*}
    \ens E_\intr^L&=\set{(i,j)\in\intr\times\intr\colon j\le i}
       \cup \set{(i,j)\in(\intr\setminus\ens L)\times\intr\colon j=i+1},\\
    \ens E_\intr^U&=\ens E_\intr^L
        \cup\set{(i,j)\in\ens L\times\ens L\colon i<j}.
  \end{align*}
\end{theorem}
\begin{proof}
  Let $\ens E_\out$, $\ens E_\inn$ and $\ens E_\ext$ be given.
  Suppose $\ens E_\intr$ is such that $\prv^T\vec e_\intr$ is maximal under
  Assumption~A.

  Firstly, by Proposition~\ref{prop:add-link} and
  since every node \new{of} $\intr$ has an access to $\ext$,
  every node $i\in\intr$ links to every node $j\in\intr$ such that $\vec v_j\ge \vec v_i$
  (see Figure~\ref{fig:intr-opt-fig1}),
  that is
  \begin{equation}\label{eq:thm-all-pred}
    \set{(i,j)\in\ens E_\intr\colon \vec v_i\le \vec v_j}
    =\set{(i,j)\in\intr\times\intr\colon \vec v_i\le \vec v_j}.
  \end{equation}
  \begin{figure}[!htb]
  \centering
  {\includegraphics{figuresthms.3}}
    \caption{Every~$i\in\intr$ must link to every~$j\in\intr$ 
    \new{with} $\vec v_j\ge\vec v_i$.}
    \label{fig:intr-opt-fig1}
  \end{figure}

  Secondly, let $(k,i)\in\ens E_\intr$ such that $k\neq i$ and $k\in\intr\setminus\ens L$.
  Let us prove that, if the node $i$ has an access to $\ext$
  by a path $\langle i,i_1,\dots,i_s\rangle$ such that
  $i_j\neq k$ for all $j=1,\dots,s$ and $i_s\in\ext$, then $\vec v_i<\vec v_k$
  (see Figure~\ref{fig:intr-opt-fig3}).
  \begin{figure}[!htb]
  \centering
  {\includegraphics{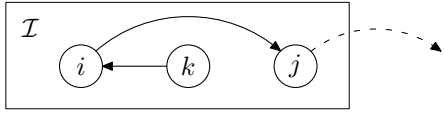}}
    \caption{The node $i$ can not have an access to $\ext$ without crossing $k$
    since in this case we should then have $\vec v_{i}<\vec v_k$.}
    \label{fig:intr-opt-fig3}
  \end{figure}
  Indeed, if we had $\vec v_k\le \vec v_i$ then, by Lemma~\ref{lem:minmaxv}$(c)$,
  there would exists $\ell\in\intr$ such that $(k,\ell)\in\ens E_\intr$ and
  $\vec v_\ell=\min_{j\gets k}\vec v_j<\vec v_i\le\vec v_k$.
  But, with $\widetilde{\ens E}_\intr=\ens E_\intr\setminus\set{(k,\ell)}$, we would have
  $\prvt^T\vec e_\intr>\prv^T\vec e_\intr$ by Proposition~\ref{prop:remove-link}
  while Assumption~A remains satisfied since the node $k$ would keep
  access to $\ext$ via the node $i$ (see Figure~\ref{fig:intr-opt-fig2}).
  \begin{figure}[!htb]
  \centering
  {\includegraphics{figuresthms.4}}%
    \caption{If $\vec v_\ell=\min_{j\ch{k}}\vec v_j$,
      then $\prvt^T\vec e_\intr>\prv^T\vec e_\intr$
      with $\widetilde{\ens E}_\out=\ens E_\out\setminus\set{(k,\ell)}$.}
    \label{fig:intr-opt-fig2}
  \end{figure}
  That contradicts the optimality \new{assumption}.
  This leads us to the conclusion that
  $\vec v_k>\vec v_i$ for every $k\in\intr\setminus\ens L$ and $i\in\ens L$.
  Moreover $\vec v_i\neq \vec v_k$ for every $i,k\in\intr\setminus\ens L$, $i\neq k$.
    Indeed, if we had $\vec v_i=\vec v_k$,
    then $(k,i)\in\ens E_\intr$ by~\eqref{eq:thm-all-pred}
    while by Lemma~\ref{lem:decreasing-path},
    the node $i$ would have an access to $\ext$ by a path independant from $k$.
    So we should have $\vec v_i<\vec v_k$.

  We conclude from this that we can relabel the nodes of $\ens N$ such that
  $\intr=\set{1,2,\dots n_\intr}$, $\ens L=\set{n_\intr-n_{\ens L}+1,\dots,n_\intr}$ and
  \begin{equation}\label{eq:vI>vF}
     \vec v_1>\vec v_2>\cdots>\vec v_{{n_\intr}-n_{\ens L}}>
     \vec v_{{n_\intr}-n_{\ens L}+1}\ge\cdots\ge\vec v_{{n_\intr}}.
  \end{equation}

  It follows also that, for $i\in\intr\setminus\ens L$ and $j>i$,
  $(i,j)\in\ens E_\intr \text{ if and only if } j=i+1$.
  Indeed, suppose first $i<n_\intr-n_{\ens L}$.
  Then, we cannot have $(i,j)\in\ens E_\intr$ with $j>i+1$
  since in this case we would contradict the ordering of the nodes given by equation~\eqref{eq:vI>vF}
  (see Figure~\ref{fig:intr-opt-fig3} again with $k=i+1$ and
  remember that by Lemma~\ref{lem:decreasing-path},
  node~$j$ has an access to~$\ext$ by a decreasing path).
  Moreover,
  node $i$ must link to some node $j>i$ in order
  to satisfy Assumption~A, so $(i,i+1)$ must belong to $\ens E_\intr$.
  Now, consider the case $i=n_\intr-n_{\ens L}$.
  Suppose we had $(i,j)\in\ens E_\intr$ with $j>i+1$.
  Let us first note that there can not exist
  two or more different links $(i,\ell)$ with $\ell\in\ens L$
  since in this case we could remove one of these links and increase strictly the PageRank
  of the set $\intr$. If $\vec v_j=\vec v_{i+1}$, we could relabel the nodes by permuting these
  two indices. If $\vec v_j<\vec v_{i+1}$, then with
  $\widetilde{\ens E}_\intr=\ens E_\intr\cup\set{(i,i+1)}\setminus\set{(i,j)}$,
  we would have $\prvt^T\vec e_\intr>\prv^T\vec e_\intr$ by Theorem~\ref{thm:p>p-d>0}
  while Assumption~A remains satisfied since the $i$ would keep
  access to $\ext$ via node $i+1$. That contradicts the optimality \new{assumption}.
  So we have proved that
  \begin{equation}\label{eq:thm-one-succ}
    \set{(i,j)\in\ens E_\intr\colon i<j \text{ and } i\in\intr\setminus\ens L}
    =\set{(i,i+1)\colon i\in\intr\setminus\ens L}.
  \end{equation}

  Thirdly, it is obvious that
    \begin{equation}\label{eq:thm-F-F}
    \set{(i,j)\in\ens E_\intr\colon i<j \text{ and } i\in\ens L}
    \subseteq\set{(i,j)\in\ens L\times\ens L\colon i<j}.
  \end{equation}

  The announced structure for a set $\ens E_\intr$ giving
  a maximal PageRank score $\prv^T\vec e_\intr$
  under Assumption~A
  now follows directly from equations~\eqref{eq:thm-all-pred},~\eqref{eq:thm-one-succ}
  and~\eqref{eq:thm-F-F}.
\end{proof}
\begin{example}\label{ex:intr-opt-ex}
\new{Let us consider the graphs given in Figure~\ref{fig:intr-opt-ex}.
  For both cases,
  the external outlink structure $\ens E_\out$ with two leaking nodes,
  as well as $\ens E_\inn$ and $\ens E_\ext$ are given.
  With $c=0.85$ and $\vec z$ the uniform probability vector,
  the optimal internal link structure for configuration~(a) is given by $\ens E_\intr=\ens E_\intr^L$,
  while in configuration~(b) we have $\ens E_\intr=\ens E_\intr^U$
  (bold arrows), with $\ens E_\intr^L$ and $\ens E_\intr^U$ defined in Theorem~\ref{thm:intr-opt}.}
  \begin{figure}[!htb]
  \centering
    \subfloat[][]{\includegraphics{figuresopts.3}}%
  \qquad\qquad
  \subfloat[][]{\includegraphics{figuresopts.2}}%
    \caption{Bold arrows represent optimal \emph{internal} link structures.
    In~(a) we have $\ens E_\intr=\ens E_\intr^L$, while $\ens E_\intr=\ens E_\intr^U$ in~(b).}
    \label{fig:intr-opt-ex}
  \end{figure}
\end{example}

Finally, combining the optimal outlink structure and the optimal internal link structure described in Theorems~\ref{thm:out-opt} and~\ref{thm:intr-opt}, we find the \emph{optimal linkage strategy} for a set of webpages. Let us note that, since we have here control on both $\ens E_\intr$ and $\ens E_\out$, there \new{are} no more cases of several final classes or several leaking nodes to consider. For an example of optimal link structure, see Figure~\ref{fig:optimal-link-structure}.
\begin{theorem}[Optimal link structure]\label{thm:website-opt}
  Let $\ens E_\inn$ and $\ens E_\ext$ be given.
  Let $\ens E_\intr$ and $\ens E_\out$ such that $\prv^T\vec e_\intr$
  is maximal under Assumption~A.
  Then there exists a permutation of the indices such that $\intr=\set{1,2,\dots,n_\intr}$,
  \[
    \vec v_1>\cdots>\vec v_{{n_\intr}}>\vec v_{{n_\intr}+1}\ge\cdots\ge\vec v_n,
  \]
  and $\ens E_\intr$ and $\ens E_\out$ have the following structure:
  \begin{align*}
    \ens E_\intr&=\set{(i,j)\in\intr\times\intr\colon j\le i \text{ or } j=i+1}, \\
    \ens E_\out&=\set{(n_\intr,n_\intr+1)}.
  \end{align*}
\end{theorem}
\begin{proof}
  Let $\ens E_\inn$ and $\ens E_\ext$ be given and suppose
  $\ens E_\intr$ and $\ens E_\out$ are such that $\prv^T\vec e_\intr$
  is maximal under Assumption~A.
  \new{Let us relabel the nodes of $\ens N$ such that $\intr=\set{1,2,\dots,n_\intr}$ and
  $\vec v_1\ge\cdots\ge\vec v_{n_\intr}>\vec v_{n_\intr+1}=\max_{j\in\ext}\vec v_j$.
  By Theorem~\ref{thm:intr-opt},
  $(i,j)\in\ens E_\intr$ for every nodes $i,j\in\intr$ such that $j\le i$.
  In particular, every node \new{of} $\intr$ has an access to node~$1$. Therefore, there is a unique
  final class $\ens F_1\subseteq\intr$ in the subgraph $(\intr,\ens E_\intr)$.
  So, by Theorem~\ref{thm:out-opt}, $\ens E_\out=\set{(k,\ell)}$
  for some $k\in\ens F_1$ and $\ell\in\ext$.
  Without loss of generality, we can suppose that $\ell=n_\intr+1$.
  By Theorem~\ref{thm:intr-opt} again, the leaking node $k=n_\intr$}
  and therefore $(i,i+1)\in\ens E_\intr$ for every node $i\in\set{1,\dots,n_\intr-1}$.
\end{proof}

Let us note that having a structure like described in Theorem~\ref{thm:website-opt} is a \emph{necessary but not sufficient} condition in order to have a maximal PageRank.
\begin{example}\label{ex:optimal-link-structure-conterexample}
  Let us show by an example that the graph structure given
  in Theorem~\ref{thm:website-opt}
  is not sufficient to have a maximal PageRank.
  Consider for instance the graphs in Figure~\ref{fig:optimal-link-structure-conterexample}.
  Let $c=0.85$ and a uniform personalization vector $\vec z = \frac{1}{n}\vun$.
  Both graphs have the link structure required Theorem~\ref{thm:website-opt} in
  order to have a maximal PageRank, with
  $\vec v_{(a)}=\begin{pmatrix}6.484&6.42&6.224&5.457\end{pmatrix}^T$
  and
  $\vec v_{(b)}=\begin{pmatrix}6.432&6.494&6.247&5.52\end{pmatrix}^T$.
  But the configuration~(a) is not optimal since in this
  case, the PageRank $\prv_{(a)}^T\vec e_\intr = 0.922$ is strictly less than the PageRank
  $\prv_{(b)}^T\vec e_\intr = 0.926$ obtained by the configuration~(b).
  Let us nevertheless note that, with a non uniform personalization vector
  $\vec z=\begin{pmatrix}0.7&0.1&0.1&0.1\end{pmatrix}^T$,
  the link structure~(a) would be optimal.
  \begin{figure}[!htb]
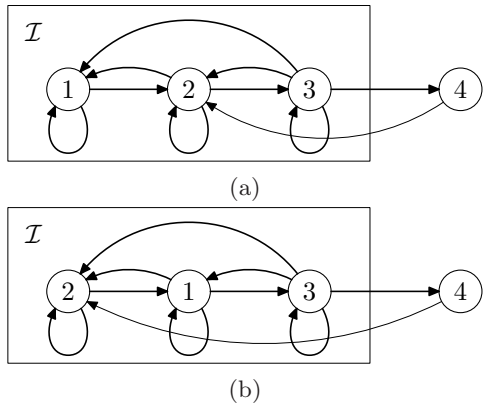

  \centering
  \subfloat[][]{\includegraphics{figures.5}}%
  \qquad\qquad
  \subfloat[][]{\includegraphics{figures.6}}%
    \caption{For $\intr=\set{1,2,3}$, $c=0.85$ and $\vec z$ uniform,
    the link structure in~(a) is not optimal \new{and yet}
    it satisfies the necessary conditions of Theorem~\ref{thm:website-opt}
    (see Example~\ref{ex:optimal-link-structure-conterexample}).}
    \label{fig:optimal-link-structure-conterexample}
  \end{figure}
\end{example}

\section{Extensions and variants}\label{sec:extensions-variants}
Let us now present some extensions and variants of the results of the previous section. We will first emphasize the role of parents of $\intr$.  Secondly, we will briefly talk about Avrachenkov--Litvak's optimal link structure for the case where $\intr$ is a singleton.  Then we will give variants of Theorem~\ref{thm:website-opt} when self-links are forbidden or when a minimal number of external outlinks is required. Finally, we will make some comments of the influence of external \emph{inlinks} on the PageRank of~$\intr$.

\subsection{Linking to parents}
If some node of $\intr$ has at least one parent in $\ext$ then the optimal linkage strategy for $\intr$ is to have an internal link structure like described in Theorem~\ref{thm:website-opt} together with {a single link to one of the parents} \new{of}~$\intr$.
\begin{corollary}[Necessity of linking to parents]
  Let $\ens E_\inn\new{\neq\emptyset}$ and $\ens E_\ext$ be given.
  Let $\ens E_\intr$ and $\ens E_\out$ such that $\prv^T\vec e_\intr$
  is maximal under Assumption~A.
  Then $\ens E_\out = \set{(i,j)}$,
  for some $i\in\intr$ and $j\in\ext$ such that $(j,k)\in\ens E_\inn$ for some $k\in\intr$.
\end{corollary}
\begin{proof}
  This is a direct consequence of Lemma~\ref{lem:vmax-parents} and Theorem~\ref{thm:website-opt}.
\end{proof}
Let us nevertheless remember that not every parent of nodes of $\intr$ will give an optimal link structure, as we have already discussed in Example~\ref{ex:parents-not-sufficient} and we develop now.
\begin{example}\label{ex:parents-not-sufficient-single}
Let us continue Example~\ref{ex:parents-not-sufficient}. We consider the graph in Figure~\ref{fig:parents-not-sufficient} as basic absorbing graph for $\intr=\set{1}$, that is $\ens E_\inn$ and $\ens E_\ext$ are given. We take $c=0.85$ as damping factor and a uniform personalization vector $\vec z=\frac{1}{n}\vun$. We have seen in Example~\ref{ex:parents-not-sufficient} than $\maxV_0=\set{2,3,4}$. Let us consider the value of the PageRank $\prv_1$ for different sets $\ens E_\intr$ and $\ens E_\out$:
\begin{center}
\begin{tabular}{c|ccccc}
  \multicolumn{6}{c}{$\quad\quad\quad\quad\quad\ens E_\out$}\\
  &$\emptyset$&$\set{(1,2)}$& $\set{(1,5)}$&$\set{(1,6)}$&$\set{(1,2),(1,3)}$\\
  \hline
  $\ens E_\intr=\emptyset$&$\diagup$&$0.1739$&$0.1402$&$0.1392$&$0.1739$\\
  $\ens E_\intr=\set{(1,1)}$&$0.5150$&$0.2600$&$0.2204$&$0.2192$&$0.2231$
\end{tabular}
\end{center}
As expected from Corollary~\ref{cor:opt-single}, the optimal linkage strategy for $\intr=\set{1}$ is to have a self-link and a link to one of the nodes $2$, $3$ or $4$. We note also that a link to node $6$, which is a parent of node $1$ provides a lower PageRank that a link to node $5$, which is not parent of $1$.
Finally, if we suppose self-links are forbidden (see below), then the optimal linkage strategy is to link to one \emph{or more} of the nodes $2,3,4$.
\end{example}

In the case where no node of $\intr$ has a parent in $\ext$, then \emph{every} structure  like described in Theorem~\ref{thm:website-opt} will give an optimal link structure.
\begin{proposition}[No external parent]
  Let $\ens E_\inn$ and $\ens E_\ext$ be given.
  Suppose that $\ens E_\inn=\emptyset$.
  Then the PageRank $\prv^T\vec e_\intr$ is maximal under Assumption~A
  if and only if
  \begin{align*}
    \ens E_\intr&=\set{(i,j)\in\intr\times\intr\colon j\le i \text{ or } j=i+1}, \\
    \ens E_\out&=\set{(n_\intr,n_\intr+1)}.
  \end{align*}
  for some permutation of the indices such that $\intr=\set{1,2,\dots,n_\intr}$.
\end{proposition}
\begin{proof}
  This follows directly from $\prv^T\vec e_\intr=(1-c)\vec z^T\vec v$ and the fact that,
  if $\ens E_\inn=\emptyset$,
  \[
    \vec v=(I-cP)^{-1}\vec e_\intr=
    \begin{pmatrix}(I-cP_\intr)^{-1}\vun\\0\end{pmatrix},
  \]
  up to a permutation of the indices.
\end{proof}

\subsection{Optimal linkage strategy for a singleton}
The optimal outlink structure for a {single webpage} has already been given by Avrachenkov and Litvak in~\cite{AL06}. Their result becomes a particular case of Theorem~\ref{thm:website-opt}. Note that in the case of a single node, the possible choices for $\ens E_\out$ can be found a~priori by considering the basic absorbing graph, since $\maxV=\maxV_0$.
\begin{corollary}[Optimal link structure for a single node]\label{cor:opt-single}
  Let $\intr=\set{i}$ and let $\ens E_\inn$ and $\ens E_\ext$ be given.
  Then the PageRank $\prv_i$
  is maximal under Assumption~A
  if and only if $\ens E_\intr=\set{(i,i)}$ and $\ens E_\out=\set{(i,j)}$
  for some $j\in\maxV_0$.
\end{corollary}
\begin{proof}
  This follows directly from Lemma~\ref{lem:V=V0-single} and Theorem~\ref{thm:website-opt}.
\end{proof}

\subsection{Optimal linkage strategy under additional assumptions}  Let us consider the problem of maximizing the PageRank $\prv^T\vec e_\intr$ when \emph{self-links are forbidden}. Indeed, it seems to be often supposed that Google's PageRank algorithm does not take self-links \new{into} account. In this case, Theorem~\ref{thm:website-opt} can be adapted readily for the case where $\abs{\intr}\ge2$.  When $\intr$ is a singleton, we must have $\ens E_\intr=\emptyset$, so $\ens E_\out$ can contain \emph{several} links, as stated in Theorem~\ref{thm:out-opt}.
\begin{corollary}[Optimal link structure with no self-links]
  Suppose $\abs{\intr}\ge2$.
  Let $\ens E_\inn$ and $\ens E_\ext$ be given.
  Let $\ens E_\intr$ and $\ens E_\out$ such that $\prv^T\vec e_\intr$
  is maximal under Assumption~A and assumption that
  there does not exist $i\in\intr$ such that $\set{(i,i)}\in\ens E_\intr$.
  Then there exists a permutation of the indices such that $\intr=\set{1,2,\dots,n_\intr}$,
    $\vec v_1>\cdots>\vec v_{{n_\intr}}>\vec v_{{n_\intr}+1}\ge\cdots\ge\vec v_n$,
  and $\ens E_\intr$ and $\ens E_\out$ have the following structure:
  \begin{align*}
    \ens E_\intr&=\set{(i,j)\in\intr\times\intr\colon j< i \text{ or } j=i+1}, \\
    \ens E_\out&=\set{(n_\intr,n_\intr+1)}.
  \end{align*}
\end{corollary}

\begin{corollary}[Optimal link structure for a single node with no self-link]
  Suppose $\intr=\set{i}$.
  Let $\ens E_\inn$ and $\ens E_\ext$ be given.
  Suppose $\ens E_\intr=\emptyset$.
  Then the PageRank $\prv_i$
  is maximal under Assumption~A
  if and only if $\emptyset\neq\ens E_\out\subseteq\maxV_0$.
\end{corollary}

Let us now consider the problem of maximizing the PageRank $\prv^T\vec e_\intr$ when \emph{several external outlinks are required}.  Then the proof of Theorem~\ref{thm:out-opt} can be adapted readily in order to have the following variant of Theorem~\ref{thm:website-opt}.
\begin{corollary}[Optimal link structure with several external outlinks]
  Let $\ens E_\inn$ and $\ens E_\ext$ be given.
  Let $\ens E_\intr$ and $\ens E_\out$ such that $\prv^T\vec e_\intr$
  is maximal under Assumption~A and assumption that
  $\abs{\ens E_\out}\ge r$.
  Then there exists a permutation of the indices such that $\intr=\set{1,2,\dots,n_\intr}$,
    $\vec v_1>\cdots>\vec v_{{n_\intr}}>\vec v_{{n_\intr}+1}\ge\cdots\ge\vec v_n$,
  and $\ens E_\intr$ and $\ens E_\out$ have the following structure:
  \begin{align*}
    \ens E_\intr&=\set{(i,j)\in\intr\times\intr\colon j< i \text{ or } j=i+1}, \\
    \ens E_\out&=\set{(n_\intr,j_k)\colon j_k\in\maxV \text{ for } k=1,\dots, r}.
  \end{align*}
\end{corollary}

\subsection{External inlinks} Finally, let us make some comments about the addition of {external inlinks} to the set~$\intr$.  It is well known that adding an inlink to a particular page always increases the PageRank of this page~\cite{AL04,IW06}. This can be viewed as a direct consequence of Theorem~\ref{thm:p>p-d>0} and Lemma~\ref{lem:minmaxv}. The case of a set of several pages $\intr$ is not so simple. We prove in the following theorem that, if the set~$\intr$ has a link structure as described in Theorem~\ref{thm:website-opt} then adding an inlink to a page of~$\intr$ from a page $j\in\ext$ which is \emph{not} a parent of some node of~$\intr$ will increase the PageRank of~$\intr$. But in general, adding an inlink to some page of~$\intr$ from $\ext$ may \emph{decrease} the PageRank of the set~$\intr$, as shown in Examples~\ref{ex:add-external-inlinks-conterexample1} and~\ref{ex:add-external-inlinks-conterexample2}.
\begin{theorem}[External inlinks]
  Let $\intr\subseteq\ens N$ and a graph defined by a set of links~$\ens E$.
  If
  \[ \min_{i\in\intr}\vec v_i>\max_{j\notin\intr}\vec v_j, \]
  then,
  for every $j\in\ext$ \new{which is not a parent of $\intr$},
  and for every $i\in\intr$,
  the graph defined by $\widetilde{\ens E}=\ens E\cup\set{(j,i)}$ gives
  $\prvt^T\vec e_\intr>\prv^T\vec e_\intr$.
\end{theorem}
\begin{proof}
  This follows directly from Theorem~\ref{thm:p>p-d>0}.
\end{proof}
\begin{example}\label{ex:add-external-inlinks-conterexample1}
  Let us show by an example that a new external inlink is not always profitable for a set~$\intr$ in
  order to improve its PageRank, even if~$\intr$ has an optimal linkage strategy.
  Consider for instance the graph in Figure~\ref{fig:add-external-inlinks-conterexample1}.
  With $c=0.85$ and $\vec z$ uniform, we have $\prv^T\vec e_\intr=0.8481$.
  But if we consider the graph defined by $\widetilde{\ens E}_\inn={\ens E}_\inn\cup\set{(3,2)}$,
  then we have $\prvt^T\vec e_\intr=0.8321<\prv^T\vec e_\intr$.
  \begin{figure}[!htb]
  \centering
  {\includegraphics{figures.7}}%
    \caption{For $\intr=\set{1,2}$,
    adding the external inlink~$(3,2)$ gives $\prvt^T\vec e_\intr<\prv^T\vec e_\intr$
    (see Example~\ref{ex:add-external-inlinks-conterexample1}).}
    \label{fig:add-external-inlinks-conterexample1}
  \end{figure}
\end{example}
\begin{example}\label{ex:add-external-inlinks-conterexample2}
  A new external inlink does not not always increase the PageRank of a set~$\intr$ in
  even if this new inlink comes from a page which is not already a parent of some node of~$\intr$.
  Consider for instance the graph in Figure~\ref{fig:add-external-inlinks-conterexample2}.
  With $c=0.85$ and $\vec z$ uniform, we have $\prv^T\vec e_\intr=0.6$.
  But if we consider the graph defined by $\widetilde{\ens E}_\inn={\ens E}_\inn\cup\set{(4,3)}$,
  then we have $\prvt^T\vec e_\intr=0.5897<\prv^T\vec e_\intr$.
  \begin{figure}[!htb]
  \centering
  {\includegraphics{figures.8}}%
    \caption{For $\intr=\set{1,2,3}$,
    adding the external inlink~$(4,3)$ gives $\prvt^T\vec e_\intr<\prv^T\vec e_\intr$
    (see Example~\ref{ex:add-external-inlinks-conterexample2}).}
    \label{fig:add-external-inlinks-conterexample2}
  \end{figure}
\end{example}

\section{Conclusions}

In this paper we provide the general shape of an optimal link structure for a website in order to maximize its PageRank. This structure with a forward chain and every possible backward links may be not intuitive. At our knowledge, it \new{has never been mentioned, while topologies} like a clique, a ring or a star are considered in the literature on collusion and alliance between pages~\cite{BYCL05,GGM05}. Moreover, this optimal structure gives new insight into the affirmation of Bianchini et al.~\cite{BGS05} that, in order to maximize the PageRank of a website, hyperlinks to the rest of the webgraph ``should be in pages with a small PageRank and that have many internal hyperlinks''. More precisely, we have seen that the leaking pages must be choosen with respect to the mean number of visits before zapping they give to the website, rather than their PageRank.

\medskip
Let us now present some possible directions for future work.

We have noticed in Example~\ref{ex:optimal-link-structure-conterexample} that the first node of $\intr$ in the forward chain of an optimal link structure is not \new{necessarily} \emph{a child} of some node of $\ext$. In the example we gave, the personalization vector was not uniform. We wonder if this could occur with a uniform personalization vector and make the following conjecture.
\begin{conjecture}
  Let $\ens E_\inn\new{\neq\emptyset}$ and $\ens E_\ext$ be given.
  Let $\ens E_\intr$ and $\ens E_\out$ such that $\prv^T\vec e_\intr$ is maximal
  under Assumption~A.
  If $\vec z=\frac{1}{n}\vun$, then there exists $j\in\ext$ such that $(j,i)\in\ens E_\inn$,
  where $i\in\argmax_{k}\vec v_k$.
\end{conjecture}
If this conjecture was true we could also ask if the node $j\in\ext$ such that $(j,i)\in\ens E_\inn$ where $i\in\argmax_{k}\vec v_k$ belongs to $\maxV$.

\medskip
Another question concerns the optimal linkage strategy in order to maximize an arbitrary linear combination of the PageRanks of the nodes of~$\intr$. In particular, we could want to maximize the PageRank $\prv^T\vec e_{\ens S}$ of \emph{a target subset} $\ens S\subseteq\intr$ by choosing $\ens E_\intr$ and $\ens E_\out$ as usual. A general shape for an optimal link structure seems difficult to find, as shown in the following example. 
\begin{example}\label{ex:target-set-examples}
  Consider the graphs in Figure~\ref{fig:target-set-examples}.
  In both cases, let $c=0.85$ and $\vec z=\frac{1}{n}\vun$.
  Let $\intr=\set{1,2,3}$ and let $\ens S=\set{1,2}$ be the target set.
  In the configuration~(a), the optimal sets of links $\ens E_\intr$ and $\ens E_\out$
  for maximizing $\prv^T\vec e_\ens S$ has the link structure
  described in Theorem~\ref{thm:website-opt}. But in~(a), the optimal $\ens E_\intr$ and $\ens E_\out$
  do not have this structure. Let us note nevertheless that, by Theorem~\ref{thm:website-opt},
  the subsets $\ens E_{\ens S}$ and $\ens E_{\mathrm{out}(\ens S)}$ must have the link structure
  described in Theorem~\ref{thm:website-opt}.
  \begin{figure}[!htb]
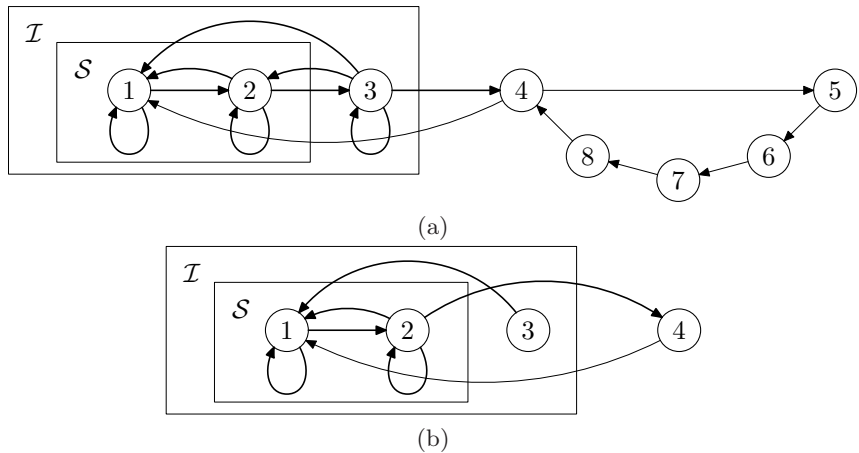

  \centering
  \subfloat[][]{\includegraphics{figures.10}}%
  \qquad\qquad
  \subfloat[][]{\includegraphics{figures.9}}%
    \caption{In~(a) and~(b), bold arrows represent
    optimal link structures for $\intr=\set{1,2,3}$ with respect to a target set $\ens S=\set{1,2}$
    (see Example~\ref{ex:target-set-examples}).}
    \label{fig:target-set-examples}
  \end{figure}
\end{example}

\section*{Acknowledgements}
This paper presents research supported by a grant ``Actions de recherche concert\'ees -- Large Graphs and Networks'' of the ``Communaut\'e Fran\c caise de Belgique'' and by the Belgian Network DYSCO (Dynamical Systems, Control, and Optimization), funded by the Interuniversity Attraction Poles Programme, initiated by the Belgian State, Science Policy Office.
The second author was supported by a research fellow grant of the ``Fonds de la Recherche Scientifique -- FNRS'' (Belgium).
The scientific responsibility rests with the authors.

\providecommand{\bysame}{\leavevmode\hbox to3em{\hrulefill}\thinspace}
\providecommand{\href}[2]{#2}

\end{document}